\documentclass[journal]{IEEEtran}
\usepackage{amsmath}
\usepackage{amsthm}
\usepackage{amsfonts}
\usepackage{amssymb}
\usepackage{graphicx}
\usepackage{url}
\usepackage{cite}
\usepackage[ruled]{algorithm2e}
\usepackage[bookmarks=false]{hyperref}
\usepackage{cryptocode}
\usepackage{booktabs}
\usepackage{multirow}
\usepackage{paralist}
\usepackage{caption}
\usepackage{mathptmx}
\usepackage[font={small,it},skip=3pt,belowskip=-2pt]{caption}
\usepackage{wrapfig}

\usepackage{cleveref}
\newcounter{centredequ}
\newenvironment{centredequ}{\refstepcounter{equation}\hfill\begin{math}}{\end{math}\hfill$(\theequation)$\par\noindent}
    \crefname{centredequ}{eq.}{eqs. }
    \Crefname{centredequ}{Eq.}{Eqs. }

\setdefaultleftmargin{7pt}{}{}{}{}{}

\newtheorem{prop}{Lemma}
\newtheorem{theorem}{Theorem}

\newcommand{\name}{PDID\xspace}

\newcommand{\names}{PDIDs\xspace}
\newcommand{\nameslong}{Password-authenticated Decentralized Identities\xspace}
\newcommand{\myparagraph}[1]{\noindent{\bf #1}}

\renewcommand{\tablename}{Tab.}

\renewcommand{\figurename}{Fig.}

\title{\nameslong}
\author{
    Pawel Szalachowski\\
    Singapore University of Technology and Design
}
\date{}

\begin{document}

% NDSS
% \IEEEoverridecommandlockouts
% \makeatletter\def\@IEEEpubidpullup{6.5\baselineskip}\makeatother
% \IEEEpubid{\parbox{\columnwidth}{
%     Network and Distributed Systems Security (NDSS) Symposium 2020\\
%     23-26 February 2020, San Diego, CA, USA\\
%     ISBN 1-891562-61-4\\
%     https://dx.doi.org/10.14722/ndss.2020.23xxx\\
%     www.ndss-symposium.org
% }
% \hspace{\columnsep}\makebox[\columnwidth]{}}
% endNDSS

\maketitle

\begin{abstract}
    Password-authenticated identities, where users establish username-password
    pairs with individual servers and use them later on for authentication, is
    the most widespread user authentication method over the Internet.  Although
    they are simple, user-friendly, and broadly adopted, they offer insecure
    authentication and position server operators as trusted parties, giving them
    full control over users' identities.  To mitigate these limitations, many
    identity systems have embraced public-key cryptography and the concept of
    decentralization.  All these systems; however, require users to create and
    manage public-private keypairs.  Unfortunately, users usually do not have
    the required knowledge and resources to properly handle cryptographic
    secrets, which arguably contributed to the failures of many end-user
    public-key infrastructures (PKIs).  In fact, as of today, no end-user PKI,
    able to authenticate users to web servers, has a significant adoption rate.
    
    In this paper, we propose \nameslong (\names), an identity and
    authentication framework where users can register their self-sovereign
    username-password pairs and use them as universal credentials.  Our system
    provides a global namespace, human-meaningful usernames, and resilience
    against username collision attacks.  A user's identity can be used to
    authenticate the user to any server without revealing that server
    anything about the password, such that no offline dictionary attacks are
    possible against the password.  We analyze \names and implement it using
    existing infrastructures and tools. We report on our implementation and
    evaluation.

\end{abstract}

\section{Introduction}
\label{sec:intro}
 Passwords have a particularly long history as a means of authenticating users to
 computer systems~\cite{bonneau2015passwords}.  Despite their
 inherent limitations and drawbacks, they are surprisingly robust to any
 techniques that try to disrupt them~\cite{bonneau2012quest}.  
%Despite all drawbacks that passwords introduce, they dominate as the default
%user authentication method in computer
%systems~\cite{bonneau2012quest,bonneau2015passwords}.  
User identities are
usually expressed as user-selected usernames and are
authenticated with passwords.  Usernames and their corresponding
password-related information are shared with a server upon registration. Then to
authenticate, a user sends its username-password pair to the server which checks
whether the pair matches the registered record.  Such identities are local (to
the server), but single sign-on systems, such as
OpenID~\cite{recordon2006openid}, extend them allowing a registered identity
to be reused `globally' for authenticating to other servers without  a need of
creating a new identity.  Although convenient for users, such
identities have significant limitations.  Most importantly, they are controlled
by their providers (i.e., the servers that have registered them).  Thus, a
user who should be an owner of its identity has to trust that the server
operator manages (and will manage) the identity appropriately.  Moreover,
systems like OpenID, allow identity providers to undermine users' privacy by
learning which websites and when users are connecting to.  Finally, the
currently dominating password-based authentication method requires users to
send their passwords in plaintext, making them  prone to various attacks.

To provide better security guarantees and enable new applications (like
signatures), public-key infrastructures (PKIs) were introduces, where
trusted authorities verify identities and assert bindings between them and their
public keys in digital certificates~\cite{cooper2008internet}.  Identities in
these systems are human-meaningful and global (usually based upon DNS), but
their security relies on a set of globally trusted authorities and the security
of the namespace they express identities in (i.e., DNS).  In past, we have
witnessed multiple attacks on authorities that resulted in impersonation attacks
on high-profile websites~\cite{leavitt2011internet}, where a trusted authority
could easily `collide' an identity by simply creating a new certificate.  To
eliminate globally trusted authorities, the idea of distributed PKIs was
presented~\cite{rivest1996sdsi}. In so called, web-of-trust
PKIs~\cite{abdul1997pgp} users create peer-to-peer trust assertions and make
trust decisions basing on them.  An important disadvantage of distributed PKIs
is that they either still rely on DNS or express identities in local namespaces, thus
cannot be used universally.  Self-certifying
identifiers~\cite{mazieres1999separating} propose names that are
cryptographically-derived from public keys. Such a namespace is global and
secure, but generated names are pseudorandom, thus it is difficult to memorize
and use them by human beings. 
The limitations of these systems led to an observation, referred to
as Zooko's trilemma~\cite{wilcox2001names}, and a related informal conjecture
that no naming system can simultaneously provide human-meaningful, global, and
secure names.  Although naming systems built upon blockchain
platforms~\cite{swartz2011squaring,loibl2014namecoin} seem to refute this
conjecture, they require to associate names with public keys.  Thus, similar
to other PKIs, they rather target servers that, unlike end-users,  are
capable to manage their cryptographic keys.

In this work, we make the following contributions.
We propose \nameslong (\name), a system that removes the above limitations.
Up to our best knowledge, it is the first identity and authentication framework
which allows users to establish human-meaningful and global
password-authenticated identities that are also resilient to collision attacks.
We instantiate \names with a combination of a blockchain platform offering
confidential smart contracts and a modified
password-authenticate key exchange protocol allowing users to use their passwords for authentication.  We present \names in the client-server
setting, where a user authenticates to the server with its username-password
pair, but the scheme can be extended to other models and applications.  We
discuss the security of our framework and present its implementation and
evaluation which indicate the feasibility of our protocol (the most common operation requires around 20 ms plus network latency).

\section{Background and Preliminaries}
\label{sec:pre}
\label{sec:pre:crypto}
\label{sec:pre:blockchain}
\label{sec:pre:pass}
% In this section, we introduce the used notation and cryptographic tools
% (\autoref{sec:pre:crypto}), briefly describe password authentication
% and a protocol we base upon (\autoref{sec:pre:pass}), as well as we introduce the
% concept of confidential smart contracts (\autoref{sec:pre:crypto}).

\newcommand{\fk}[2]{\ensuremath{\mathsf{f}_{#1}(#2)}}
\newcommand{\getsR}{\ensuremath{\xleftarrow{R}}}
\newcommand{\Hash}[1]{\ensuremath{\mathsf{H}(#1)}}
\newcommand{\HashPrim}[1]{\ensuremath{\mathsf{H'}(#1)}}
\newcommand{\AEnc}[2]{\ensuremath{\mathsf{AEnc}_{#1}(#2)}}
\newcommand{\ADec}[2]{\ensuremath{\mathsf{ADec}_{#1}(#2)}}
\newcommand{\Gen}{\ensuremath{\mathsf{Gen}()}}
\newcommand{\PEnc}[2]{\ensuremath{\mathsf{PEnc}_{#1}(#2)}}
\newcommand{\PDec}[2]{\ensuremath{\mathsf{PDec}_{#1}(#2)}}
\newcommand{\Zq}[2]{\ensuremath{Z_q}}

\myparagraph{Notation and Cryptography}
Throughout the paper we use the following notation
% \begin{inparaenum}[]
\begin{compactitem}
	\item $G$ denotes a finite cyclic group of order $q$ with a generator $g\in G$;
% \smallskip
	\item $r\getsR S$  denotes that $r$ is an element randomly selected
		from the set $S$;
% \smallskip
	\item \fk{k}{m} is a keyed-pseudorandom function that for key $k$ and
		message $m$ outputs a pseudorandom string from $\{0,1\}^n$;
% \smallskip
	\item \Hash{m} and \HashPrim{m} are cryptographic hash functions
		that for message $m$ output values from $\{0,1\}^n$ and $G$, respectively;
% \smallskip
	\item \AEnc{k}{m} is an encryption algorithm of an authenticated
		encryption scheme, that for key $k$ and message $m$ outputs the
		corresponding ciphertext $c$;
% \smallskip
    \item \ADec{k}{c} is the corresponding decryption algorithm,
        decrypting the message $m$ from the ciphertext $c$ given the key
        $k$, or failing with incorrect input;
% \smallskip
    \item \Gen\xspace is a public-key generation algorithm, returning a
        private-public keypair $\langle sk,pk\rangle$;
% \smallskip
	\item \PEnc{pk}{m} is a public-key encryption algorithm, that produces
		ciphertext $c$ for the given public key $pk$ and message $m$;
% \smallskip
	\item \PDec{sk}{c} is the corresponding public-key decryption algorithm,
		recovering the message $m$ given the ciphertext $c$ and the
        corresponding secret key $sk$, or failing with incorrect input.
\end{compactitem}

\myparagraph{Password Authentication}
% Passwords have a particularly long history as a means of authenticating users to
% computer systems~\cite{bonneau2015passwords}.  Despite their
% inherent limitations and drawbacks, they are surprisingly robust to any
% techniques that try to disrupt them~\cite{bonneau2012quest}.  
Password-based
authentication on the Internet is dominated by the following method (or its
slight modification):
\begin{compactitem}
    \item \textit{Registration.} The user registers its identity by providing,
        via a secure channel, a username $U$ and a password $pwd$ to the server.
        The server selects a random salt $s\getsR{\{0,1\}}^n$ and
        stores the mapping 
	$U: \langle s, \Hash{s, pwd}\rangle.$
        %\footnotetext{The purpose of adding a random salt is to
        %mitigate offline dictionary attacks possible after the server's storage
        %is compromised (i.e., with randomized password hashes, an adversary
        %would need to precompute a dictionary per salt, what is memory- and
        %computational-hard). Although other constructions are possible, the
        %combination of salt and hashed concatenation of salt and password is the
        %most common one.} 
    \item \textit{Authentication.} To authenticate the user sends
        its username-password pair $\langle U, pwd'\rangle$ to the server, which
        identifies the mapping, and checks if
		$\Hash{s, pwd'}\overset{?}{=}\Hash{s, pwd}.$
\end{compactitem}

Despite its popularity and wide-spread adoption, this protocol has a major flaw
since in every authentication the password is sent in plaintext.  This limitation
requires a secure channel between the parties for each authentication, but even
then, it makes passwords vulnerable to multiple attack vectors (like server-side
malware).
To address this limitation, Bellovin and Merritt~\cite{10.5555/882488.884178}
proposed the first password-authenticated key exchange (PAKE) protocol, where
two parties can securely establish a high-entropy secret key from the memorable
password they share (with an already establish secret key,
the user can easily authenticate to the server).
Since then, there have been proposed multiple PAKE protocols with various
efficiency and security
properties; %~\cite{wu1998secure,boyko2000provably,katz2001efficient,hao2008password}.
% Most of those protocols do not allow an adversary to learn the password or any
% information allowing her to run offline dictionary attacks against the password.
however, most of them require either to send salt in plaintext (facilitating offline
precomputation attacks) %\footnote{An adversary, knowing a username-salt pair, can
% prepare a dedicated dictionary targeting this user that can be used immediately
% after compromising the server.} 
or to store effective
passwords by servers (allowing adversaries to instantly compromise all passwords
after the server's compromise).  Only recently, Jarecki et al. proposed
OPAQUE~\cite{jarecki2018opaque}, a PAKE protocol that removes these issues and
introduces low transmission and computation overheads.
OPAQUE bases on the oblivious pseudorandom function (OPRF) defined as
\newcommand{\Fk}[2]{\ensuremath{\mathsf{F}_{#1}(#2)}}
\begin{centredequ}
    \label{eq:oprf}
\Fk{k}{m} = \Hash{m, (\HashPrim{m})^k}.
\end{centredequ}

To register, a user sends to the server its username-password pair $\langle U,
pwd\rangle$.
% (Only this transmission requires a secure channel between the parties.)
The server, after receiving the request, computes the following
% \begin{gather*}
$k_s\getsR\Zq\xspace{}; k\gets\Fk{k_s}{pwd}; p_s\getsR\Zq\xspace{}; P_s\gets g^{p_s}; %\\
p_u\getsR\Zq\xspace{}; P_u\gets g^{p_u}; c\gets\AEnc{k}{p_u,P_u,P_s};$
% \end{gather*}
and saves $\langle k_s,p_s,P_s,P_u,c\rangle$ as the \textit{password metadata}
corresponding to the username $U$.
After the registration is complete, the user can use its credentials to
authenticate to the server. In order to do so, the user computes
% \begin{gather*}
$r\getsR\Zq\xspace{};\alpha\gets(\HashPrim{pwd})^r;
x_u\gets\Zq\xspace{};X_u\gets g^{x_u};$
% \end{gather*}
and sends $U, \alpha, X_u$ to the server.
Upon receiving this message, the server computes
$
x_s\getsR\Zq\xspace{};X_s\gets g^{x_s}; \beta \gets \alpha^{k_s};
$
and sends $\beta, X_s, c$ back to the user who then computes
$
k\gets\Hash{pwd, \beta^{1/r}};
\langle p_u, P_u, P_s\rangle\gets\ADec{k}{c}.
$
Now, the user (with $p_u, x_u, P_s, X_s$), and the server (with
$p_s, x_s, P_u, X_u$) can run a key exchange protocol, like
HMQV~\cite{krawczyk2005hmqv}, to establish a shared key.

\myparagraph{Confidential Smart Contracts}
Blockchain platforms, initiated by Bitcoin~\cite{nakamoto2019bitcoin}, combine
append-only cryptographic data structures with a distributed consensus
algorithm.  They allow to build highly-available, censorship-resistant,
transparent, and verifiable systems minimizing trust in third parties.
Initially, proposed for peer-to-peer payments, blockchain platforms were
subsequently 
enriched by smart contracts, which allow anyone to deploy code with any
(implementable) logic
and interact with this code
over the blockchain platform.
However, the transparency of those systems can also be seen as an important
disadvantage % . Recording all (plaintext) transactions in a public ledger,
% inevitably limits the applicability of those platforms to 
for privacy-demanding
applications and users. 
It also turned out that 
% The community seeing this limitation quickly proposed
% private payment-specific blockchains, however, 
designing a platform for
confidential smart contracts turned out to be a more challenging task. Only
recently we have witnessed some promising proposals providing confidentiality in
smart contracts via
cryptographic tools (like commitments, multi-party computation,
or zero-knowledge computation integrity proofs). % , to achieve confidentiality in
% smart
% contracts. %~\cite{kosba2016hawk,bowe2018zexe,lu2019honeybadgermpc,benhamouda2020can}.
Unfortunately, as of today, these platforms introduce significant
efficiency bottlenecks, effectively prohibiting the deployment of sophisticated
smart contracts.

Blockchain
platforms that leverage a trusted execution environment
(TEE) are solutions with a more practical focus. %~\cite{yuan2018shadoweth,cheng2019ekiden,brandenburger2018blockchain,bowman2018private,russinovich2019ccf,homoliak2020aquareum}.
The TEE technology, usually realized with Intel SGX, a) allows to run
code within \textit{secure enclaves} which cannot be compromised even by the system
operator, b) facilitate \textit{remote attestation}, able to prove which enclave code is
running on a remote machine, and c) offer \textit{sealing} which enables
enclaves to encrypt their secret data and to deposit it on untrusted storage
This
toolset %, together with a consensus protocol (guaranteeing the consistent view of
% a global distributed ledger), 
allowed to provide high-performant and
confidential smart contracts, where contracts, their data, and all
transactions are confidential, despite being ordered, validated, and recorded on
a public ledger.  For instance, Brandenburger et
al.~\cite{brandenburger2018blockchain} extend Hyperledger
Fabric~\cite{androulaki2018hyperledger} by confidential smart contracts.
Their architecture distinguishes two types of enclaves: one responsible for
verifying the blockchain state integrity, and another for executing actual smart
contracts confidentially.  Nodes can join the system via remote
attestation and signing up to a special enclave registry.  For each created
smart contract, a dedicated keypair is generated and published. This keypair is
used to protect the contract and uniquely identify it. 
%
%Examples of other projects which develop confidential and public smart contract
%platforms include OASIS, TEEX, and
%Enigma\footnote{\url{https://www.oasislabs.com/}; \url{https://teex.io/}; \url{https://www.enigma.co/}}.

Due to  practical reasons, we instantiate \names with a
platform that bases on the TEE assumption; however, with the continuing progress of
platforms basing upon cryptographic assumptions,  we do not see major obstacles
in implementing \names with such a platform.  We assume that the platform
exposes public keys (e.g., as described previously) which allow users to
interact with platform contracts confidentially, by sending encrypted
transactions.  We do not assume a specific consensus protocol, but we require
that the platform allows to generate inclusion proofs for appended transactions.
% In platforms leveraging a traditional Byzantine consensus, such a proof can be
% expressed as a multisignature of $f+1$ nodes (where $f$ is the maximum number of
% malicious nodes -- the standard assumption is
% that to tolerate $f$ malicious nodes, the total number of nodes has to be at
% least $3f+1$.), while in longest-chain protocols it is usually a
% Merkle path proof, rooted in a block belonging to the longest-chain.

\section{The \name Framework}
\label{sec:overview}
% In this section, we present \names. We start with the problem formulation
% (\autoref{sec:overview:prob}), then we give intuitions and design rationale behind
% our framework (\autoref{sec:overview:intuitions}), followed by design details
% (\autoref{sec:overview:details}), and a discussion on \names management
% (\autoref{sec:overview:mgmt}).

\subsection{Problem Formulation}
\label{sec:overview:prob}
Our goal is to propose a user identity and authentication framework. 
Although we present our system in the client-server model where only users are
authenticated, it can be adjusted to other models and authentication scenarios
(e.g., mutual authentication). %\footnote{In practice, thanks to the increasing
% HTTPS adoption, TLS authentication of servers have become a de facto standard in
% recent years.}
We introduce the following parties.
%\begin{compactdesc}
\textbf{User} is a human being that wishes to use a service that requires
	authentication. The user inputs the service name $S$, as well as its
	username and password credentials $\langle U, pwd\rangle$, used for his
	authentication.  The user operates its client software to execute the
	actual authentication protocol.  % For a simple description, we represent
	% the user and his client software as a single entity (i.e., user).
\textbf{Server} represents the service the user wants to use and which requires
	authentication.  The server participates in the authentication protocol
	and aims to verify the user's credentials (i.e., whether
	the user is an owner of the identity he claims).
%\end{compactdesc}
%
We assume that the protocol parties have access to a blockchain platform
with confidential smart contracts (see \autoref{sec:pre:blockchain}).
For our framework, we seek the following properties.

\begin{compactdesc}
    \item[Human-meaningful Names:] identifiers are memorable by human
	    beings, such that users do not any need special infrastructures or
		devices to remember them. Ideally, they are user-selected usernames.
    \item[Global Namespace:] identifiers resolve to the same identity no
	    matter when and where they are being resolved. It
		guarantees that identities can be used universally.
    \item[(Collision-)Secure Names:] identities cannot be impersonated by
        forging identifiers (e.g., by hijacking or creating a new identity with
        the same identifier). In practice, this property requires that there is
        no trusted authority(ies) privileged to manage identities.  For
        instance, in authority-based PKIs, an authority can simply impersonate
        an identity by creating a malicious certificate claiming the same
        identifier, while in OpenID the identity provider can freely modify or
        use stored user records. 
	\item[Memorable Secrets:] users can use memorable secrets (i.e.,
		passwords) for authenticating their identities. This property
		is desired due to the popularity and advantages of passwords in
		user authentication.  It also enables to design an identity
		system which is seamless to end-users.
    \item[Secure Authentication:] the user authentication process does not
        reveal, even to the verifying server, the user's password or any
        information allowing to run offline dictionary attacks against the
        password.
\end{compactdesc}

The first three properties constitute Zooko's trilemma.  Although some PKI
systems refute the trilemma, they operate on public-private keypairs.  This may be
acceptable for servers able to manage their keys and certificates, but
it may be too demanding for users who prefer to authenticate using  username-password pairs.
Therefore, we introduce the fourth requirement on memorable secrets which also
allows to avoid user-side changes.  %Next, we require secure
% authentication, what for our aimed universal authentication system also extends
% to servers which should not be able to learn (even via offline attacks) users' passwords.
In addition to secure authentication, we also require that the system is
efficient (i.e., does not introduce prohibitive overheads), and  keeps the
users' privacy on the same (or a similar) level as in today's authentication.

We assume an adversary whose goal is to authenticate on behalf of the user or to
learn the user's password.  We assume that the adversary cannot compromise
users' passwords, the used cryptographic primitives and protocols, and the
deployed blockchain platform.  The adversary can compromise a server to which the
user authenticates, and in this case, the adversary aims to attack the user's
password.  We require that the system not only protects from revealing users'
passwords but in particular, should not reveal to the adversary any information
which would enable her to run offline attacks on passwords.  We require that the
smart contract execution and the consensus protocol are secure, although we assume
that the adversary can compromise up to a tolerable number of blockchain nodes
(e.g., up to $1/3$ of all nodes in Byzantine consensus).  We assume that the
adversary operating such a compromised node may be interested in attacking the
system properties (e.g., attempting to run offline attacks).  Side channel
attacks, like timing, power, or cache attacks, are out of the scope of our
adversary model.

\subsection{Intuitions and Design Rationale}
\label{sec:overview:intuitions}

%\begin{figure}[t!]
%  \centering
%  \includegraphics[width=.8\linewidth]{fig/overview}
%  \caption{A high-level overview of the naive approach.}
%  \label{fig:overview}
%\end{figure}

%Designing an identity framework with the stated properties is challenging.
%Systems that attempt to solve the trilemma, usually, propose a PKI variant,
%where bindings between user-selected names and public keys are recorded on a
%public distributed ledger (the append-only ledger and a large-scale consensus
%guarantee that the names are unique and secure).  Such a design is seemingly
%contradicting when applied to a password system.  Obviously, passwords or even
%their salted hashes cannot be stored publicly, but on the other hand,
%human-meaningful names to be collision-secure seem to require some kind of
%`global coordination'.  

To illustrate our design process better, in this
section, we first consider a naive approach to the problem. %, depicted in \autoref{fig:overview}.
In this protocol, we
introduce a trusted and highly available entity called the \textit{Global
Password Manager} (GPM). The GPM is responsible for handling identity
registrations, keeping all username and password pairs, and for assisting
servers with user authentication requests.
    To participate in the protocol, the user selects its username and
        password and registers this pair with the GPM, % \footnote{We assume that
        % the parties exchange their messages using a secure channel although it
        % is not important for our discussion.} 
        which ensures that the username is unique
        and saves the credentials in its database.%  The registration
        % step is one-time per identity.
    After the identity is established, the user can use it for 
        authentication, using the two-step protocol: 
        \begin{inparaenum}[a)]
            \item To authenticate, the user sends its credentials to
                a targeted server.
	    \item The server, to verify whether the credentials are valid,
		    contacts the GPM which
			checks and notifies the server whether the username-password pair is recorded in
			its database and notifies.
			Depending on the outcome, the server either successfully
			authenticates the user or terminates the protocol.
        \end{inparaenum}

With our assumption about the GPM, the protocol satisfies some of our challenging
requirements.  With a single trusted GPM who manages its local credentials
database, the protocol guarantees that identities are global and unique, and kept private.
% (i.e., the GPM makes sure that no username is equivocated or manipulated).
Moreover, the credentials are universal, since a user registering once with the GPM,
can use its username and password to authenticate with any server supporting the
protocol.
Unfortunately, such a simple approach has two following fundamental issues.

Firstly, realizing such a trustworthy centralized GPM would be difficult in
practice.  Centralized systems are single points of failure (in terms of
security, privacy, and availability), introduce higher censorship risks, and can
be manipulated easier.  These limitations make a centralized GPM an unacceptable
design, especially in the context of universal and global identities.
GPM could be implemented using a highly-available infrastructure, like cloud computing, but then the system would be prone to censorship by the infrastructure's operator.
Therefore, one of our design decision is to implement the GPM's functionality as
a confidential smart contract (see \autoref{sec:pre:blockchain}) executed over a
decentralized blockchain platform.  In such a setting, the GPM is replaced by a
smart contract, keeping and managing credentials according to the rules
specified by the code. The system would also benefit from the blockchain
properties, providing verifiability and distributed control (mitigating
censorship), while keeping the state and the execution of the GPM smart contract
confidential.

Another major drawback of the naive protocol is that the server learns the
user's credentials.  This limitation is quite standard in centralized identity
systems (where a user and server share the user's effective password). However, with
universal decentralized identities it is unacceptable
since otherwise, only one malicious server could compromise universal credentials
which could be used for authentication to any other servers.  Therefore, our
goal is to realize the authentication process, such that the user can authenticate
to any server, but without revealing to the server any information allowing to learn the password.
To realize it, we extend the OPAQUE protocol 
to the three-party setting, where the
protocol is run between the user, the server, and the GPM.

\subsection{Details}
\label{sec:overview:details}
\begin{figure}[t!]
  \centering
  \includegraphics[width=\linewidth]{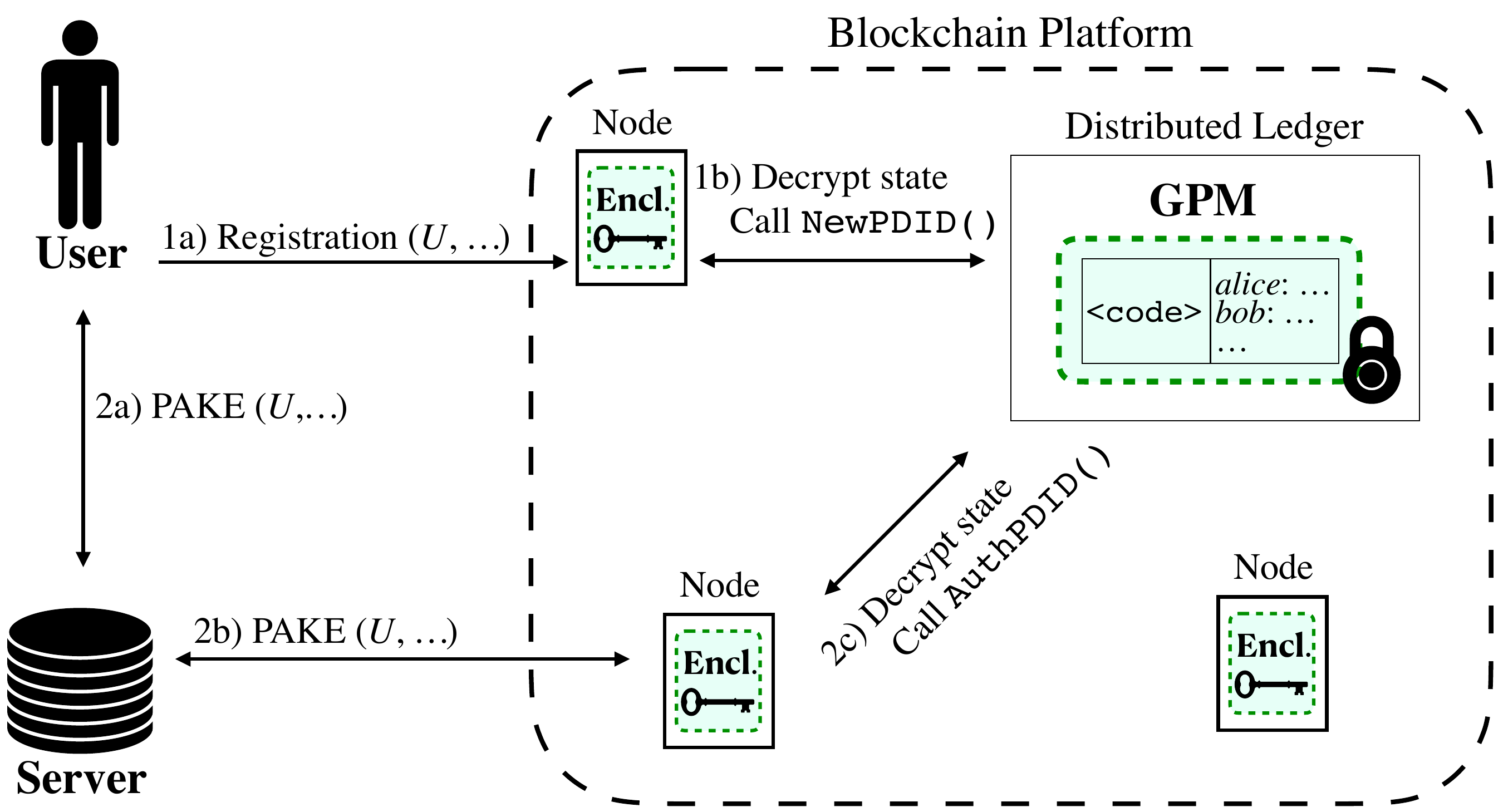}
  \caption{A high-level overview of the \name framework.}
  \label{fig:overview2}
\end{figure}

In \autoref{fig:overview2}, we show a high-level overview of our framework
instantiated with a TEE-based blockchain platform offering confidential smart
contracts. The GPM is implemented as a
smart contract with the encrypted state which can be accessed and modified only by trusted
enclaves.  A GPM instance is created before the protocol's deployment, and on
its creation, it is assigned with a unique blockchain's public-private keypair (see
\autoref{sec:pre:blockchain}), denoted as
\begin{centredequ}
    \label{eq:keys}
\langle pk_b, sk_b\rangle.
\end{centredequ}
GPM can be created by anyone, as long as its address and code (to potentially audit it) are publicly known and agreed upon. When created, GPM does not have to be maintained.
Interactions with the GPM are conducted via transactions that are sent to
(untrusted) blockchain nodes running (trusted) enclaves that execute the GPM's
code.  The confidentiality of these transactions is protected by public-key
encryption using the public key of the GPM instance (i.e., $pk_b$), and users
and servers are preloaded with this key.  The GPM consists of two main methods,
for \name registration and authentication, and the $users[]$ dictionary which
maps usernames to their password metadata (the dictionary is empty upon the
contract creation).

In the following, we describe the registration and authentication procedures.
For a simple description, in our protocols and pseudocodes, we omit some basic
sanity checks like parsing, checking whether received elements belong to the
group $G$, or decryption failures.  We emphasize; however, that if an error
occurs at one of those, the party should terminate the protocol in the failure
mode.

\subsubsection{Registration}
\begin{figure}[t!]
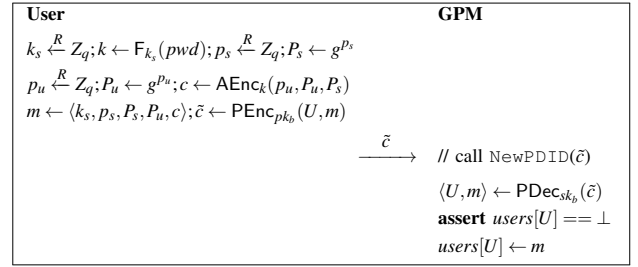

  \centering
	\begin{center}\scalebox{0.8}{\fbox{ \pseudocode{%
\textbf{User} \> \> \textbf{GPM} \\
k_s\getsR\Zq\xspace{}; k\gets\Fk{k_s}{pwd}; p_s\getsR\Zq\xspace{}; P_s\gets g^{p_s} \> \> \\
p_u\getsR\Zq\xspace{}; P_u\gets g^{p_u}; c\gets\AEnc{k}{p_u,P_u,P_s} \> \> \\
m\gets\langle k_s,p_s,P_s,P_u,c\rangle;\tilde{c}\gets\PEnc{pk_b}{U, m}\> \> \\
\> \sendmessagerightx[.8cm]{0}{\tilde{c}} \> \text{// call \texttt{NewPDID}($\tilde{c}$)} \\ 
\> \> \langle U, m\rangle \gets \PDec{sk_b}{\tilde{c}}\\
\> \> \textbf{assert } users[U] == \bot \\
\> \> users[U] \gets m
} } }
\end{center}
\caption{The registration process, where the user inputs its username $U$
    and password $pwd$, and $pk_b$ is the blockchain's public key.}
\label{fig:reg}
    \vspace{-0.5cm}
\end{figure}
The registration process is depicted in \autoref{fig:reg}.
To register its identity (i.e., \name), the user first prepares the password metadata which is computed using his
password $pwd$.  In OPAQUE, the
metadata is computed as shown in \autoref{fig:reg}, using OPRF (see
\autoref{eq:oprf}), modular exponentiations of secret values, and authenticated
encryption. The password metadata $m$ is:
\begin{centredequ}
    \label{eq:meta}
    \langle k_s,p_s,P_s,P_u,c\rangle.
\end{centredequ}
Originally, in OPAQUE, this phase is run by an entity storing the password
metadata (the GPM in our case). In the \name framework, it is generated on the
user's side since 
\begin{inparaenum}[a)]
    \item the process requires a good source of entropy which smart contracts,
        being fully deterministic, cannot themself provide,
    \item similarly, TEEs able to provide randomness, like Intel SGX, have been
        demonstrated to do it unreliably~\cite{aumasson2016sgx}, and
    \item in this setting, only the user knows the password $pwd$, thus,
        even in the case of a catastrophic attack (like compromised $sk_b$),
        the adversary learns the only one-way transformation of the
        password and not $pwd$ itself,
    \item the complexity of the GPM's code is minimized.
\end{inparaenum}

After the metadata $m$ is created, it is accompanied with the username $U$,
encrypted under the blockchain public key $pk_b$ as the ciphertext $\tilde{c}$,
and sent to the blockchain platform as a transaction triggering the registration
method of the GPM.  After the transaction is appended to the ledger, a
blockchain node, noticing the request, restores the GPM's code and state, and
calls its \texttt{NewPDID()} method. As presented in \autoref{fig:reg}, the code
first ensures that the username $U$ is not registered yet and then assigns the
password metadata to the username in the $users[]$ dictionary.  At this point,
the user's \name is established and ready for being used in the authentication
process.
(We note that even though the GPM's state is encrypted, storing the password
metadata instead of plain passwords, gives an additional level of security, since
even with a catastrophic event, like compromised $sk_b$ or the GPM, the
adversary still needs to run dictionary attacks again every single password.)

\subsubsection{Authentication}
\label{sec:overview:auth}
After a \name is registered, the user should be able to use its credentials.
We require that 
\begin{inparaenum}[a)]
    \item the user is able to use its credentials $\langle U, pwd\rangle$ to
        authenticate to any server (supporting the scheme), 
    \item the server is able to verify that the user knows the password
        corresponding to its claimed identity $U$,
\item the server does not learn $pwd$ or any information enabling to recover it (e.g., via offline
	dictionary attacks).
\end{inparaenum}

The last two requirements may seem contradictory since the server needs to
verify the identity without possessing its corresponding password metadata.  In
the \name framework, servers indeed do not store password metadata, which
instead is stored only as part of the GPM's confidential state. Then, 
to satisfy these requirements, we extend the OPAQUE authentication protocol to
the three-party setting, where the GPM's trusted code assists the server in
verifying the user's credentials.  When abstracting the server and the GPM as a
single entity, they essentially execute the server's side of the OPAQUE
authentication; however, since the server alone does not have any
password-related information it has to communicate with the GPM. The details of
the \name authentication process are presented in \autoref{fig:auth} and described below.

\begin{figure*}[t!]
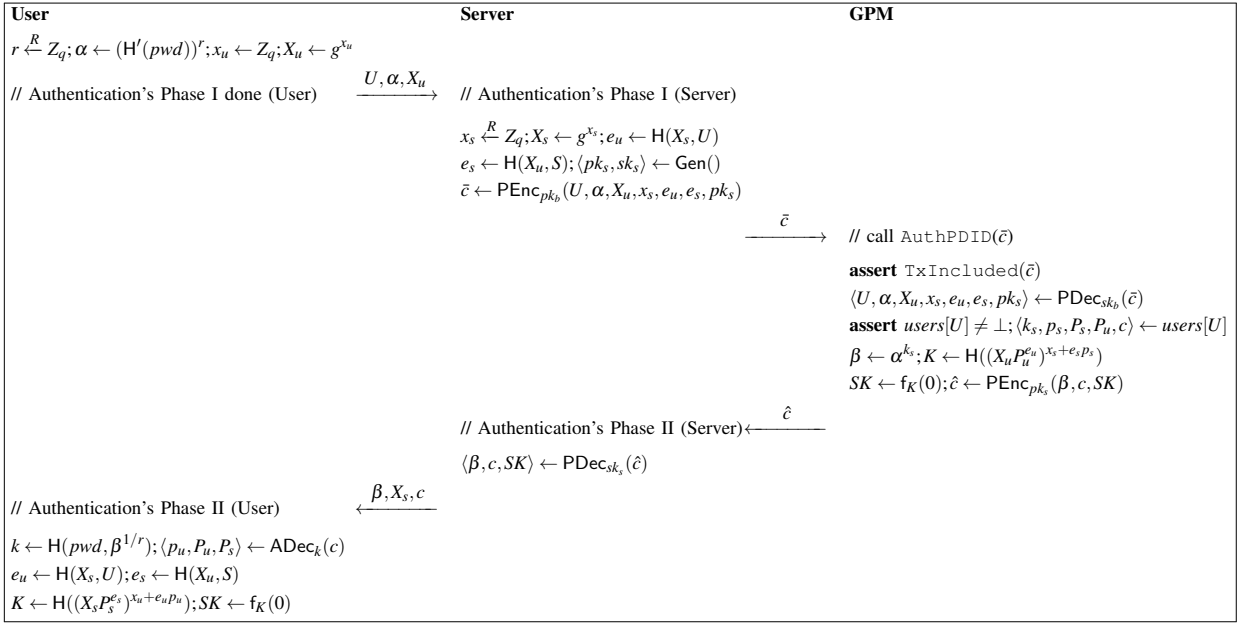

    \begin{center} \scalebox{0.8}{\fbox{\pseudocode{%
\textbf{User} \< \< \textbf{Server} \< \< \textbf{GPM} \\ 
	r\getsR\Zq\xspace{};\alpha\gets(\HashPrim{pwd})^r;x_u\gets\Zq\xspace{};X_u\gets g^{x_u} \< \< \< \< \\
	\text{// Authentication's Phase I done (User)} \< \sendmessagerightx[1.2cm]{0}{U, \alpha, X_u} \< \text{// Authentication's Phase I (Server)} \< \< \\
	\< \< x_s\getsR\Zq\xspace{};X_s\gets g^{x_s}; e_u\gets\Hash{X_s,U} \< \< \\
	\< \< e_s\gets\Hash{X_u,S};\langle pk_s, sk_s\rangle\gets\Gen \< \< \\
	\< \< \bar{c}\gets\PEnc{pk_b}{U, \alpha, X_u, x_s, e_u, e_s, pk_s} \< \< \\
	\< \< \< \sendmessagerightx[1.2cm]{0}{\bar{c}} \< \text{// call \texttt{AuthPDID}($\bar{c}$)}\\ 
    \< \< \< \< \textbf{assert } \texttt{TxIncluded}(\bar{c}) \\
	\< \< \< \< \langle U, \alpha, X_u, x_s, e_u, e_s, pk_s\rangle\gets\PDec{sk_b}{\bar{c}} \\
	\< \< \< \< \textbf{assert } users[U]\neq\bot;\langle k_s,p_s,P_s,P_u,c\rangle\gets users[U] \\
	\< \< \< \< \beta \gets \alpha^{k_s};K\gets\Hash{(X_uP_u^{e_u})^{x_s+e_sp_s}} \\
	\< \< \< \< SK\gets\fk{K}{0}; \hat{c}\gets\PEnc{pk_s}{\beta,c,SK} \\
	\< \<\text{// Authentication's Phase II (Server)} \< \sendmessageleftx[1.2cm]{0}{\hat{c}} \< \\ 
	\< \< \langle\beta,c,SK\rangle\gets\PDec{sk_s}{\hat{c}} \< \< \\
	\text{// Authentication's Phase II (User)}\< \sendmessageleftx[1.2cm]{0}{\beta, X_s, c} \< \< \< \\
	k\gets\Hash{pwd, \beta^{1/r}};\langle p_u, P_u, P_s\rangle\gets\ADec{k}{c} \< \< \< \< \\
	e_u\gets\Hash{X_s,U};e_s\gets\Hash{X_u,S} \< \< \< \< \\
	K\gets\Hash{(X_sP_s^{e_s})^{x_u+e_up_u}};SK\gets\fk{K}{0} \< \< \< \<
        }}
}
\end{center}
    \caption{The \name authentication process, where $S$ is the server's
    identity.}
    \vspace{-0.6cm}
\label{fig:auth}
\end{figure*}

\begin{compactenum}
    \item The user, as in OPAQUE, computes his
        contributions $\alpha$ and $X_u$ to the authenticated key exchange
        protocol and sends the username $U$ together with $\alpha$
        and $X_u$ to the server.  

    \item The server first generates its contribution $X_s=g^{x_s}$ to the key
        exchange.  For similar reasons % (i.e., a lack of reliable randomness on
        % smart contracts and minimizing the complexity of the GPM), 
	as in the
        registration, this step is conducted by the server and not by the GPM.
        Next, the server computes the session-unique values $e_u$ and $e_s$,
        which will be used in HMQV.  Then the server generates
        a keypair $\langle pk_s, sk_s\rangle$ that binds the
	server-GPM communication to the current session (without revealing the server's identity). % and which will be used
        % by the GPM to contact the server back securely and without revealing the
        % server's identity.  
	Then, the server encrypts (with $pk_b$) its
        transaction containing the user's input ($U$, $\alpha$, and $X_u$), the
        server's key exchange contributions ($X_s$, $x_s$, $e_u$, and $e_s$), as
        well as its name $S$ and the ephemeral public key $pk_s$.  The encrypted
        transaction $\bar{c}$ is submitted to the blockchain platform.

    \item The transaction, after being appended to the ledger, triggers the
        GPM's \texttt{AuthPDID()} method, which is executed by a blockchain
        node within its secure enclave as follows.

\begin{inparaenum}[(a)]
        \item First, it calls the \texttt{TxIncluded($\bar{c}$)} method to
            ensure that the transaction $\bar{c}$ is already
            appended in the blockchain (this check is blockchain- and
            consensus-specific -- see \autoref{sec:pre:blockchain}).  The
            purpose of this check is to eliminate offline attacks conducted by a
            malicious blockchain node
            (see \autoref{prop:off-tee} in \autoref{sec:analysis:off} for more
            details).
        \item Next, the code decrypts the ciphertext $\bar{c}$, identifies the
            user, and restores his password metadata.
        \item Then, the OPAQUE protocol is continued, computing $\beta$
            from user-provided $\alpha$ and restored ${k_s}$. (The $\beta$
            values will allow the user to restore the key $k$ which was used for
            encrypting the metadata's ciphertext $c$.)
        \item OPAQUE's final phase is to derive a key that will be shared
            between the server and the user. To accomplish it, 
            we combine OPAQUE with the
            method it uses by default, i.e., the HMQV protocol.  
            Therefore, the GPM computes  $(X_uP_u^{e_u})^{x_s+e_sp_s}$ and
            hashes it into a key $K$.
    \item The \texttt{AuthPDID()} ends its execution by deriving (from $K$) the
        shared session key $SK$, which together with the values $\beta$ and $c$
            is encrypted under $pk_s$ as the ciphertext $\hat{c}$ and sent back
            to the server.
\end{inparaenum}

\item The server decrypts $\hat{c}$, saves $SK$, and passes $\beta, X_s, c$ to
    the user, to enable him to obtain the same key $SK$.  

    \item After receiving these values, the user  continues the protocol
	    by computing $\beta^{1/r}$ (which equals $(\HashPrim{pwd})^{k_s}$),
		which hashed with the user's password $pwd$ generates the key
		$k$ under which the metadata's ciphertext $c$ is encrypted (see
		\autoref{fig:reg}).  After decrypting $c$, the user finishes
		the protocol by computing $(X_sP_s^{e_s})^{x_u+e_up_u}$ and
		$SK$.
\end{compactenum}
The protocol finishes with the parties obtaining the same shared key $SK$,
which, in addition to authentication, can be used for protecting their
subsequent communication.  To authenticate, the user can simply use $SK$ to
(encrypt and) authenticate the exchanged messages together with the first
application-layer data.

\section{Security Analysis}
\label{sec:analysis}
\label{sec:analysis:priv}
\label{sec:analysis:off}
\label{sec:analysis:name}
\label{sec:analysis:auth}
\label{sec:analysis:online}
% In this section we discuss the security of our framework.  We start with
% showing the properties of its namespace (\autoref{sec:analysis:name}), then we
% discuss the authentication security (\autoref{sec:analysis:auth}), offline
% attacks (\autoref{sec:analysis:off}), and lastly, we discuss online attacks
% (\autoref{sec:analysis:online}) and the privacy (\autoref{sec:analysis:priv}).

\subsection{Global and Collision-secure Names}
We require that the identity system provides global and collision secure names.
In this section, we show  that \names provide those properties, assuming a
secure blockchain platform.

\begin{theorem}
    \label{prop:global}
\names provide a collision-secure global namespace.
\end{theorem}
% \begin{proof}[Proof (sketch)]
\begin{proof}[sketch]
    Since we assumed that the blockchain platform deployed is secure, the GPM's
	state, at any point in time, has one canonical view which is consistent
	with all previous views.  This guarantees, that the namespace
	consisting of identifiers recorded in the GPM's $users[]$ dictionary
	represents a global view of all usernames.
    The blockchain-platform assumption also implies that the transactions are
    processed by nodes correctly, i.e., the GPM's state can only be changed by
    secure enclaves processing transactions whose order is agreed on with the
    underlying consensus algorithm.  Given that, it is easy to show that the
    namespace is collision-secure since if there are two conflicting
    transactions, trying to register the same name $U$, the latter (according to
    their execution order in the ledger) will inevitably fail, as the enclave code
    processing it will not continue the \texttt{NewPDID()} method (see
    \autoref{fig:reg}), after executing the following assertion:
    $
    \textbf{assert } users[U] == \bot.
    $
\end{proof}

\subsection{Authentication}
Another stated requirement is the security of the authentication process.
This section argues that our framework provides secure authentication for
\names.

\begin{prop}
No adversary can authenticate on behalf of the user $U$ without knowing the
corresponding password $pwd$.
\end{prop}
\begin{proof}[sketch]
    An adversary, without knowing the user's password, can authenticate as the user
    only if one of the following occurs
    \begin{inparaenum}[1)]
        \item the adversary can impersonate the user's identity $U$, registering
            and authenticating with a new $\langle U, pwd'\rangle$ pair,
        \item the adversary can compute a correct session key.
    \end{inparaenum}
    As we show in \autoref{prop:global}, the identities are global and unique,
    thus the adversary cannot impersonate the user registering another
    username-password pair, contradicting the first option.  Therefore, the only
    way to authenticate on behalf of the user is to obtain a correct session key
    shared with the server.  This option; however, can be eliminated by showing
    that with our assumptions, the authentication phase can be reduced to the
    OPAQUE protocol  and benefit from its properties (as we show in
    \autoref{prop:passive_off}).
\end{proof}

\subsection{Offline Attacks}
Besides preventing \name's impersonation, our framework aims at limiting offline
dictionary attacks,  where an adversary can gather information allowing her to
check whether the used password is in a dictionary.
In this section, we show that \names prevent such attacks.

\begin{prop}
    \label{prop:passive_off}
No passive adversary can learn the user's password $pwd$ or launch a successful
    offline attack against $pwd$.
\end{prop}
\begin{proof}[sketch]
    To prove this lemma,  we show that our construction
    can be reduced to the combination of the OPAQUE and HMQV protocols, for
    which security proofs are presented in their respective
    papers~\cite{jarecki2018opaque,krawczyk2005hmqv}.
    In our construction (see \autoref{fig:auth}), the server and the GPM
    together execute  OPAQUE's server-side authentication logic.
    Since the GPM's execution is confidential, the adversary learns only two
    additional messages (when compared with the original OPAQUE execution),
    exchanged between the server and the GPM:
    $\bar{c}\gets\PEnc{pk_b}{U, \alpha, X_u, x_s, e_u, e_s, pk_s}$
    and 
    $\hat{c}\gets\PEnc{pk_s}{\beta,c,SK}.$
    As the adversary cannot compromise the underlying cryptographic primitives,
    the corresponding secret keys $sk_b$ and $sk_s$, %\footnote{We recall that
    % $\langle pk_s, sk_s\rangle$ is a fresh ephemeral keypair.}, 
    and the
    blockchain platform, she is unable to learn anything more than from the
    original OPAQUE's execution.
\end{proof}

A malicious server is a more interesting case.  Such a server does
not have an interest in compromising the session key (since it knows it already),
but may be interested in learning passwords of users authenticating with it.
Below we argue that such a server cannot learn any information allowing it to
launch even offline password attacks.

\begin{prop}
No adversary able to compromise a server can learn the user's password $pwd$ or
	launch a successful offline attack against $pwd$.
\end{prop}
\begin{proof}[sketch]
    During the execution of the protocol (see \autoref{fig:auth}), an adversary
	controlling the server learns the user's input $U,\alpha, X_u$ and can
	freely control values $U, \alpha, X_u, x_s, e_u, e_s, pk_s$ passed to
	the GPM.  The adversary learns GPM's output $\beta, c, SK$, and the
	goal of the adversary is to find the password $pwd$ or any information
	allowing for offline attacks against $pwd$. 
According to the registration and authentication protocols, the knowledge of
$k_s$ is necessary to run offline attacks against $pwd$. With $k_s$ the
adversary could try to keep generating different symmetric keys $k'$ from
potential passwords $pwd'$:
$
k'\gets\Fk{k_s}{pwd'},
$
and keep testing them against the known ciphertext $c$, which is computed as
$
\AEnc{k}{p_u,P_u,P_s}.
$
However, $k_s$ is a high-entropy secret which is not revealed to the adversary. The adversary
can interact with the GPM, passing different $\alpha'$ values and obtaining
$
\beta'\gets\alpha'^{k_s},
$ 
but if this interaction would allow the adversary to learn $k_s$, or even
to compute $\alpha'^{k_s}$ for any non-queried $\alpha'$, that would be
equivalent with breaking the One-More Diffie-Hellman problem (assumed
to be hard by the OPAQUE protocol), contradicting our assumptions.
\end{proof}

Similarly, it is easy to show that such an adversary cannot learn any value
of the password metadata ($p_u,P_u,P_s$), and subsequently, cannot learn a
session key $SK$ for any authentication that she does not participate in.

Although we assume that the blockchain platform is secure, a tolerable number
of individual nodes (usually, up to $1/3$ of all nodes) can be compromised.
Then, a particularly interesting case is when a malicious node operator interacts
with the secure enclave it runs. Such an operator, cannot read the enclave's
memory or influence its execution steps but can interact with it offline.
In particular, the node can emulate the authentication process (see
\autoref{fig:auth}) by trying multiple passwords, calling the GPM's
\texttt{AuthPDID()} method locally, and checking if the user's
session key and the key outputted by the enclave match.
We show that the \name framework eliminates such attacks.

\begin{prop}
    \label{prop:off-tee}
	An adversary able to compromise a tolerable number of blockchain nodes cannot launch a
	successful offline dictionary attack against the user's password $pwd$.
\end{prop}
% \begin{proof}[Proof (sketch)]
\begin{proof}[sketch]
    To prove this property, we show from the GPM's construction, that its
    \texttt{AuthPDID()} method computes a shared secret key only for a
    transaction that was already added to the blockchain.
    In our construction, we use a technique similar to the one presented by
    Kaptchuk et al.~\cite{kaptchuk2019giving}, where upon receiving a
    transaction $\bar{c}$, its processing method \texttt{AuthPDID()} first calls
    $
    \textbf{assert } \texttt{TxIncluded}(\bar{c}),
    $
    which guarantees that the transaction is already part of the ledger.
    Since the adversary, even compromising a tolerable number of nodes, is not
    able to compromise the properties of the blockchain platform, she is not
    able to overcome this assertion with any transaction that is not 
    appended to the blockchain.
    With this check, the adversary controlling a compromised node and
    interacting with its trusted enclave offline, cannot emulate the user-server
    authentication, and to get any GPM's output she needs to register the
    transactions on the blockchain, making her attack attempts visible to the
    network (i.e., online).
\end{proof}

\subsection{Online Attacks}
Online password guessing attacks, where an adversary interacts with the
authentication system trying to guess correct username-password pairs, is a
generic attack against any password system.  The \name framework is not an
exception to such attacks, and an adversary can just try different passwords
interacting with supporting servers or with the GPM directly.  A popular way of
mitigating such attacks is rate limiting.  Usually, it introduces a trade-off
between the security and availability, and such  mitigation would be implementable
at the network-level or within the GPM (e.g., via small state representing the
number of recent authentication attempts).  Moreover, blockchain platforms
enable an interesting extension of this technique. Instead of limiting authentication
attempts, after a threshold number of attempts in a time window, the platform
could require a small payment that could disincentivize adversaries
from guessing passwords. We leave details
of such a solution as future work.

\subsection{Privacy}
The \name framework does not introduce any message flows or mechanisms
violating the user's privacy when compared with the traditional password-based
authentication (users interact only with servers, except for the registration).
The messages exchanged between servers and the GPM are encrypted, thus do not
reveal anything about the authenticating user.  Moreover, the server's keys used in
the GPM-server communication are ephemeral and unknown to an observer, therefore,
the observer investigating the blockchain logs would not be able to determine
the server's identity (we do not consider network-level deanonymization
attacks).  Finally, all usernames are stored as part of the encrypted state and
processed only by trusted enclaves confidentially. % thus enumerating usernames
% just by looking at the blockchain content is impossible.\footnote{For instance,
% an adversary could try to register every possible username to learn which are
% registered already.  This approach, however, could be mitigated by means
% discussed in \autoref{sec:analysis:online}.}

\section{Implementation and Evaluation}
\label{sec:implementation}
% We implemented and evaluated the \name framework to prove its feasibility.  In
% this section, we report on our implementation and present the evaluation of our
% system.

\subsection{Implementation}
\begin{table}[t]
\centering
    \caption{Performance of the operations (ms).}
\label{tab:perf}
  \includegraphics[width=\linewidth]{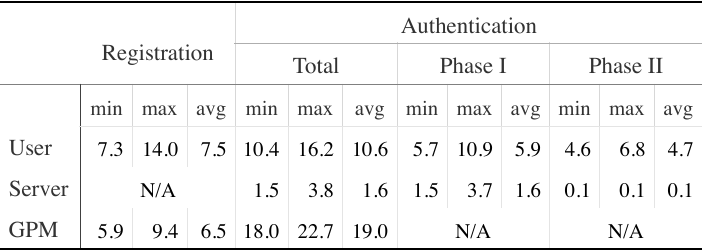}
\end{table}
We fully implemented the \name framework and our
implementation consists of a supported user's client (executing the
user's
logic as in \autoref{fig:reg} and \autoref{fig:auth}), a server (authenticating
users as in \autoref{fig:auth}), and the GPM handling registrations and assisting
servers in user authentication.  The user and server functionalities are
implemented in C. To realize the GPM, we used a recent Hyperledger Fabric Private
Chaincode (FPC) framework~\cite{brandenburger2018blockchain} which extends
Hyperledger Fabric by confidential smart contracts.  The GPM is implemented in
C++ as a smart contract of this platform by using Intel SGX SDK to run the contract
within an enclave.
For encryption and hashing, 
use the NaCl library~\cite{bernstein2009cryptography} with its defaults for the user and server
implementations, and the TweetNaCl library for the GPM.
% these operations in the GPM we used %~\cite{bernstein2014tweetnacl} 
% (a more portable version of
% NaCl). %\footnote{We emphasize that the GPM's code is executed within an SGX
% enclave which cannot use standard system libraries.} 
% We use a combination of the
% Curve25519, Salsa20, and Poly1305 cryptographic primitives (see
% details~\cite{bernstein2009cryptography}) as public-key encryption, a
% combination of Salsa20 and Poly1305 for authenticated encryption, and SHA-512 as
% a hash function.  
We implemented the modified OPAQUE and HMQV protocols with elliptic
curve cryptography for the group operations.  Our implementation bases upon and
extends the Easy-ECC library %\footnote{\url{https://github.com/esxgx/easy-ecc}}
and we used the secp256r1 curve by default.
Our code is publicly available at \url{https://github.com/pszal/pdid}.

\subsection{Evaluation}
To evaluate our \name implementation we conducted a series of experiments.
First, we evaluated computational overheads introduced by the \name framework.
We set up an FPC testbed and executed full \name registration and authentication
operations  1\,000 times each.  In every run, we measured the time required to
complete different protocol steps.  To measure execution times, we used a
commodity laptop equipped with SGX-enabled Intel i7-7600U (2.80GHz) CPU, 8GB of RAM, and run
under Linux.  In our experiments, we used a conservative setting
measuring the computational overhead of specific procedures executed sequentially on
`fresh' registration and authentication requests. We did not use any caching
strategies, parallelization, sophisticated parametrization, or
request/transaction batching, which would amortize the execution time, although
we see these techniques as desired in a deployment-ready implementation.  The
results of our experiments are presented in \autoref{tab:perf}, reporting
total execution times and times for different authentication
phases (see \autoref{fig:auth}). %  By user's `Phase I' of the
% authentication, we denote the operations with which the user initiates the
% protocol until the first message is sent to the server.  Authentication's
% `Phase I' on the server side starts when the server receives this message and
% lasts until the ciphertext $\bar{c}$ is generated, while the server's `Phase II' is
% between receiving $\hat{c}$ and sending the last protocol's message.  The
% steps executed by the user after he receives the last message are denoted as
% `Phase II' of the user's authentication. 

As presented, even in our unoptimized setting %, executed on a single core
\names introduce a small computational overhead.  The
registration on the user side requires around 7ms on average, while the
authentication process takes in total around 10ms on average, % This overhead is
%comparable with client-side TLS verifications.  
and requires  the user to keep only a 97 byte long state.  More
importantly, the server's side is even faster, requiring
only 1.63ms per authentication (in total, on average), dominated by its first phase. It
allows a server to conduct around 613 authentications per second. A server
needs to store only 64 byte long state (an ephemeral secret key $sk_s$) per authentication.  Similarly, our protocol introduces a small
transmission overhead with  message sizes are between 74 and 300 bytes.

% The throughput of the GPM depends not only on contract execution times but
% mainly on the consensus layer. % where blockchain nodes run an actual consensus
% protocol to agree on the transaction order.  
To improve performance and allow
higher flexibility, Hyperledger Fabric separates contract execution and
consensus layers, with
distinct node functions responsible for ordering and executing transactions.  Given
that, the performance of the \name framework is bounded by the consensus
layer (since only once transactions are ordered they can be executed).
Fortunately, the
performance of this layer has been extensively investigated in previous
studies, and even in large-scale distributed deployments, a Hyperledger Fabric
network yields throughput between 2\,000 and 3\,000 transactions per
second~\cite{androulaki2018hyperledger}, introducing the end-to-end latency
between 500 and 800 ms, respectively. % \footnote{Our protocol introduces small
% transmission overhead with the message sizes are between 74B and 300B.  The
% end-to-end latency, in our cases, is introduced by the server-GPM
% communication, also mainly depends on the consensus layer and network
% conditions of the blockchain platform.} 

The GPM's \texttt{NewPDID()} code, executed within an SGX
enclave, handles a  \name registration in around 6.54ms on average.  Therefore,
a single core executing the GPM's registration can handle around 153
registrations per second (around 0.55 million per hour).  This throughput seems
to be sufficient to handle even a global-scale registration load using only a
single core.  Each registered
\name requires only 260 byte password metadata in the GPM's contract
state.  Due to the elliptic curve operations, the authentication
on the GPM is relatively slower, requiring 19.00ms on average, which
yields the throughput of around 52 authentications per second on a single
core.  To shed a light
on this number, we refer to Thomas et al.~\cite{thomas2019protecting} who report
that for 670\,000 users, Google has experienced 21 million authentications
for 28 days in early 2019 with a peak 2\,192 authentications per 100 seconds (i.e.,
around 22 authentications per second).
Approximating their results, a single
core executing GPM's \texttt{AuthPDID()} can handle around 1.58 million
users in the peak and around 4 million for the
averaged load.  We note, that this throughput can be scaled horizontally
(to the limits of the consensus layer). %, thus 38 nodes (assuming
% ideal load balancing) would be able to saturate the consensus layer throughput
% at 2\,000 authentications per second.  There are possible further performance
% improvements, including
% more performant GPM implementation or scalability solutions (like sharding) of
% the platform itself.

\section{Related work}
\label{sec:related}
% The subject of identity and naming infrastructures has been extensively
% investigated over the last decades.   In this section,
% we discuss and compare only with the most related work.

\subsection{Password-authenticated Identities}
Username and password pairs, as described in \autoref{sec:pre:pass}, are
arguably the most popular credentials for user
authentication~\cite{bonneau2015passwords}.  % Users using password-authenticated
% identities, typically, establish their identities with a server first and then
% use these identities to access the server (as described in
% \autoref{sec:pre:pass}).  
Typically, the identity is expressed as a user-selected
username (i.e., login) which is local to the server and cannot be used for
authenticating to other servers.
To make such `local' logins more universal and useful, decentralized
authentication was
proposed.  For instance, with
OpenID~\cite{recordon2006openid} users can use their identities
registered with a single server (i.e., identity provider) to authenticate with other servers with no need of creating a
new dedicated identity for them.  % In these systems, to authenticate with a
%relying server, a user authenticates to its identity providers (usually with a
%username-password pair) which in turn asserts the identity's
%validity to the relying server (who trusts the identity provider in that
%scope).
%
Password-authenticated identities provide critical advantages, contributing to
their surprising domination over seemingly superior
alternatives~\cite{bonneau2012quest}.  Firstly, they are
expressed as human-meaningful names, limiting the need of additional devices
or infrastructures for storing or processing them.  Secondly, passwords, broadly
considered as memorable secrets, significantly simplify the secret managements
(especially, on the user's side).  Lastly, they have a particularly long history
as authentication means, thus they typically do not introduce any
adoption or operation overheads.  They come
with some limitations; however.  The main drawback 
is that users do not control their own identities and a malicious server
could impersonate or terminate any identities at its will.  In fact, not
only servers have to be trusted. For example, identities  expressed with e-mail
addresses rely on DNS. % which is hierarchical by design, thus for resolving a
%domain name a trust placed in its operator (as well as, in all operators of all
%`upper-domains') is required. % \footnote{For instance, a user with the e-mail
%identity \texttt{user@domain.ac.uk}, not only trusts \texttt{domain.ac.uk}, but
%the entire DNS hierarchy (i.e., \texttt{ac.uk}, \texttt{uk}, and the root).}
OpenID, although convenient, enables identity providers to learn what servers (and
when) users contact.
Lastly, most of these schemes use the standard insecure authentication
(discussed in \autoref{sec:pre:pass}). % \names eliminate these drawbacks,
% enabling decentralized and self-sovereign identities, at the same time
% providing the advantages of user-password pairs with secure authentication.

\begin{table*}[t!]
  \centering
  \caption{Comparison of different schemes in terms of the desired properties.}
  \includegraphics[width=1.0\linewidth]{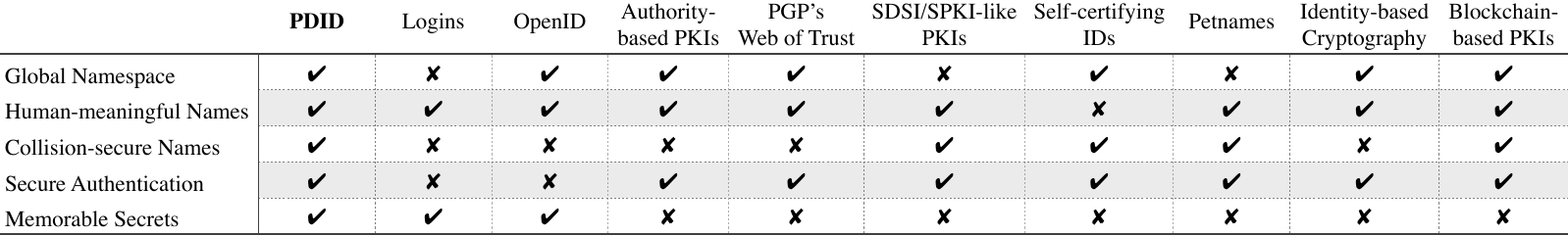}
  \label{fig:comparison}
    \vspace{-0.75cm}
\end{table*}
\subsection{Public-Key-authenticated Identities}
% Schemes that rely on public-key cryptography for authentication constitute another
% popular class of identity systems.  
%Unlike password-authenticated identities,
Systems described in this part require users to establish and manage
public-private keypairs. Although it may improve security
%facilitating secure authentication protocols) 
and enable new applications, %(like
%digital signatures), 
in practice, the public-private keypair management  has
been proved to be a challenging task, not only for users but even for allegedly
more tech-savvy server operators~\cite{krombholz2017have}. 

\subsubsection{Certificate-based}
Authority-based public-key infrastructures (PKIs), like X.509
PKI~\cite{cooper2008internet}, are designed to manage mappings between
identities and their public keys.  Typically, they introduce trusted certification
authorities (CAs) that verify bindings between identities
and their claimed public keys and assert this fact in signed
certificates.  X.509 is prominently used together with TLS for authenticating
web services (identified by domain names); however, it is also adopted for user
identities (expressed by full names
 and/or e-mail addresses)~\cite{ramsdell2004secure}. % In this case, users are issued with
% digital certificates, certifying that their identities (expressed by full names
% and/or e-mail addresses) are associated with their public keys.  
These PKIs
require trust in CAs, which usually trust DNS for identity verification, and  % The
%required trust is a major weakness of those systems and 
there have been multiple
real-world attacks on CAs reported to date~\cite{leavitt2011internet}.  Although
many recent approaches try to improve these PKIs, they usually target CAs'
accountability, transparency, and attack
detection~\cite{laurie2014certificate,basin2014arpki,ryan2014enhanced,melara2015coniks,syta2016keeping},
but without changing the fundamental CA trust assumptions.

Relaxing the assumption of trusted authorities is a design goal of
decentralized PKIs, where no trusted party is needed to verify identities,
which in turn, are verified and asserted by other system participants.  In
those systems, trust decisions are made solely by users, depending on their
trust estimation of the quality and length of `trust chains'. %  Trust relations
Such a web-of-trust model was prominently
proposed in PGP~\cite{abdul1997pgp} for securing e-mails (which still rely on
the DNS hierarchy).  SDSI/SPKI~\cite{ellison1999simple} is a
distributed PKI with local namespaces extending the web-of-trust paradigm by
introducing groups, access control, and security policies.  These systems
eliminate trusted parties, but their namespaces either allow collisions or are
local. Moreover, they
require an infrastructure for distributing trust relations. % The GNU Name
% System~\cite{wachs2014censorship} proposes such an infrastructure to be
% implemented via a distributed hash table.

Petnames~\cite{stiegler2005introduction} is an anti-phishing system where users
themselves can assign local (private) names to keys of parties they interact
with, distrusting names placed in certificates.  The intention of those local
names is that, if a certificate for a lexicographically-similar phishing website
is presented, it will not be trusted by the user by default, since the website
will not have its petname.  The system requires users to keep maintaining
correct bindings between petnames and keys, which in a dynamic or multi-key
environment, like the Internet, can be troublesome and harm 
usability~\cite{ferdous2009security}.

\subsubsection{Certificate-less}
Since certificates introduce substantial overheads and their management poses
significant challenges, especially, for security-unaware users, eliminating
digital certificates was a design goal of certificate-less systems.
Self-certifying identifiers~\cite{mazieres1999separating} base on the idea of
deriving identity directly from the public-key, usually, by simply hashing it.
Such hash-names are short (20-32 bytes), global, and collision-secure, %as their
% uniqueness is guaranteed by the properties of the hash function deployed, 
and
users can create them  by themselves. %without interacting with any other party.  
%Self-certifying
%identities are deployed, for instance, in the Tor's hidden service
%\texttt{.onion} namespace~\cite{syverson2004tor}. % or in the HIP
% protocol~\cite{moskowitz2008host}.  
Unfortunately, these names %in those systems 
are
represented  by pseudorandom strings which, despite being relatively short, are
not easily memorable by humans, thus they require a dedicated name distribution
infrastructure.

Identity-based cryptography enables to generate public-private keypairs in a
way where private keys are freely selected by users~\cite{shamir1984identity}.
Since the public keys can be human-meaningful identities themselves, these
systems do not need certificates or name discovery infrastructures. %, while
% enabling public-key encryption~\cite{boneh2001identity} or signature
% verification~\cite{hess2002efficient}, just by inputting human-friendly
%identity's names.  
The main drawback of this approach is that keypairs have to
be generated by a trusted party that learns secret keys. % Moreover, when used
% in the global environment, these schemes would need some way of coordination to
% guarantee unique names.

\subsubsection{Blockchain-based}
Blockchain platforms were early seen as promising infrastructures for
implementing distributed identities. %  They provide decentralization, verifiability,
% transparency, censorships resistance, and they heavily use unique
% addresses, i.e.,  public keys (or their hashes) of participants which can be
%used as native identities.
%
%Built upon these observations, %~\cite{swartz2011squaring},
Namecoin~\cite{loibl2014namecoin} is a blockchain-based PKI platform allowing
users to register arbitrary identities and associate them with public keys.  The
system refutes the Zookoo's conjecture, providing memorizable names associated
with key pairs, without trusted parties.  %Despite the low adoption of
% Namecoin~\cite{kalodner2015empirical}, 
The community followed the Namecoin's
design,  proposing systems with additional
features~\cite{ali2016blockstack} or extended trust
models~\cite{al2017scpki,dykcik2018blockpki}.

Baars~\cite{baars2016towards} discusses the applicability of  blockchain technology to provide self-sovereign Identities, while Goodell and Aste~\cite{goodell2019decentralized} propose an architecture where besides technical aspects they point out importance of a careful regulation in such systems. AttriChain~\cite{shao2020attrichain} is a holistic scheme using a permissioned blockchain to offer  self-sovereign identities with threshold traceability and on-chain access control. All these systems require users to create and manage cryptographic keys.

Decentralized Identifiers (DIDs)~\cite{reed2020decentralized} is an attempt to
unify the management of decentralized digital identities.  With DIDs, users can
create their self-sovereign identities (associating them with their public keys)
and anchor them with a blockchain platform their trust.  DIDs are under heavy
development, and are implemented and experimentally deployed as part of the
Hyperledger project. %~\cite{dhillon2017hyperledger}.

\subsection{Comparison}
In \autoref{fig:comparison}, we compare different name systems
with ours in terms of the desired properties.  For the row `Collision-secure
Names', we put a negative mark, if the system introduces trusted party(ies)
managing identities.  Systems  providing this property differ in
assumptions under which the property is achieved.  For instance, self-certifying identifiers
require a collision-resistant hash function, web-of-trust PKIs need a trusted `fragment' of a
peer-to-peer network, and blockchain-based systems require that a
(super)majority of a Sybil-resistant network is honest.
We also note that our instantiation of the \name framework requires the TEE
assumption. %, although \names can be instantiated with platforms not requiring
% TEE (see \autoref{sec:pre:blockchain}). 
Our comparison shows that \names is the
only system to achieve all the desired properties.

\section{Conclusions}
\label{sec:conclusions}
In this work, we presented \names, a framework providing password-authenticated
decentralized identities with global and human-meaningful names.  Up to our best
knowledge, it is the first system achieving these properties, and in comparison to the state-of-the-art systems, \name does not require users to use and manage cryptographic keys for authentication.
In our system, a
user registers his username and password-derived information with a confidential
smart contract and then can use these credentials to authenticate to any server.
For authentication, we combine our framework with the OPAQUE protocol, resulting
in an authentication system where even the server cannot learn the user's
password or any information leading to offline dictionary attacks against it.
We report on the implementation of our system and evaluation results.  In the future,
we plan to investigate \name management and extend our system beyond authentication (e.g., to
decentralized storage).

\section*{Acknowledgment}
This work is supported by A*STAR under its RIE2020 Advanced Manufacturing and Engineering (AME) Programmtic Programme  (Award A19E3b0099) and by Ministry of Education, Singapore, under its MOE AcRF Tier 2 grant (MOE2018-T2-1-111).

\bibliographystyle{unsrt}
\bibliography{ref}

% \appendix
% \input{sec/disc}

%\begin{IEEEbiography}[{\includegraphics[width=1in,height=1.25in,clip,keepaspectratio]{mshell}}]{Michael Shell}
% or if you just want to reserve a space for a photo:

% \begin{IEEEbiography}[{\includegraphics[width=1in,height=1.2in,clip,keepaspectratio]{fig/pawel}}]{Pawel Szalachowski}
%	is currently an Assistant Professor at Singapore University of Technology and Design (SUTD), Singapore. Prior to joining SUTD, he was a senior researcher at ETH Zurich. He received his PhD degree in Computer Science from Warsaw University of Technology in 2012. His current research interests include systems security and blockchains.
%\end{IEEEbiography}

\end{document}